\documentclass[reqno,11pt]{article} 
\usepackage{amsfonts,amssymb,amsmath,amsthm,mathrsfs}
\usepackage{color}
\usepackage{tikz} 
\usepackage{hyperref} 
\usepackage{caption} 

\usepackage{marvosym}

\makeatletter
\newcommand\ackname{Acknowledgements}
\if@titlepage
\newenvironment{acknowledgements}{
   \titlepage
   \null\vfil
   \@beginparpenalty\@lowpenalty
   \begin{center}
      \bfseries \ackname
      \@endparpenalty\@M
   \end{center}}
{\par\vfil\null\endtitlepage}
\else
\newenvironment{acknowledgements}{
   \if@twocolumn
   \section*{\abstractname}
   \else
   \small
   \begin{center}
      {\bfseries \ackname\vspace{-.5em}\vspace{\z@}}
   \end{center}
   \quotation
   \fi}
{\if@twocolumn\else\endquotation\fi}
\fi
\makeatother

\hypersetup{
   pdfmenubar=false,        
   pdfauthor={Eris Runa},     
   pdfsubject={Subject},   
   pdfcreator={Eris Runa},   
   pdfproducer={Eris Runa}, 
   pdfnewwindow=true,      
   colorlinks=false,       
   linkcolor=blue,          
   citecolor=blue,        
   filecolor=magenta,      
   urlcolor=cyan           
}

\usepackage[left=1in,right=1in,top=1.3in,bottom=1.4in]{geometry}

\newcommand{\kabs}[1]{\ensuremath{\vert#1\vert}}
\usepackage{xspace}
\usepackage{enumitem}

\newcommand{\Z}{\mathbb{Z}}

\newcommand{\C}{\mathbb{C}}
\newcommand{\Sec}{\ensuremath{\S}}

\newcommand{\id}{\mathrm{Id}}

\newcommand{\I}{\mathrm{Id_m}}
\newcommand{\dx}{{\,\mathrm{d}x}}

\newcommand{\dt}{{\,\mathrm{d}t}}
\newcommand{\dz}{{\,\mathrm{d}z}}

\newcommand{\N}{\mathbb{N}}
\newcommand{\R}{\mathbb{R}}

\newcommand{\T}{\mathbb{T}}
\newcommand{\Lp}{\mathrm{L}}

\newcommand{\bmo}{\mathrm{BMO}}
\newcommand\scal[2]{{\left\langle #1 ,#2\right\rangle}}
\newcommand\scalare[1]{{\left\langle #1 \right\rangle}}
\newcommand\norm[1]{\| #1\|}

\newcommand{\ottoboh}[1]{ }
\newcommand{\akmboh}[1]{{\color{green}}}
\newcommand{\BrascampLiebBoh}[1]{}

\def\d{\,\mathrm{d}}

\newcommand{\insieme}[1]{\left \{#1\right \}}

\newcommand{\eg}{e.g.,\xspace}

\newcommand{\etal}{\emph{et~al.}\xspace}

\newcommand{\ie}{i.e.,\xspace}

\DeclareMathOperator{\diver}{\nabla^{*}}
\DeclareMathOperator{\dist}{dist}

\newtheorem{theorem}{Theorem}[section]
\newtheorem{thm}[theorem]{Theorem}

\newtheorem{lemma}[theorem]{Lemma}

\newtheorem{remark}[theorem]{Remark}

\newtheorem{proposition}[theorem]{Proposition}

\newtheorem{corollary}[theorem]{Corollary}

\title{Finite range decomposition for a general  class of elliptic operators} 

\usepackage{authblk}
\author{Eris Runa\thanks{eris.runa@mis.mpg.de}}
\affil{Max Planck Institut for Mathematics in the Sciences,\\ Inselstrasse 22, Leipzig\\ Germany}

\usepackage{esint}

\newcommand{\abs}[1]{{\lvert #1\rvert}}

\newcommand{\Acal}   {{\mathcal A }}
\newcommand{\Bcal}   {{\mathcal B }}
\newcommand{\Ccal}   {{\mathcal C }} 
\newcommand{\Dcal}   {{\mathcal D }} 
\newcommand{\Ecal}   {{\mathcal E }}

\newcommand{\Lcal}   {{\mathcal L }} 
\newcommand{\Mcal}   {{\mathcal M }}

\newcommand{\Rcal}   {{\mathcal R }}

\newcommand{\cH  }{\boldsymbol{\mathcal H}}
\newcommand{\As }{\mathscr{A}}
\newcommand{\Bs }{\mathscr{B}}
\newcommand{\Ts }{\mathscr{T}}
\newcommand{\Rs }{\mathscr{R}}
\newcommand{\Cs }{\mathscr{C}}

\usepackage{parskip}
\usepackage{titlesec}
\usepackage{titling}

\date{}
\usepackage{microtype}

\begin{document}

\maketitle

\begin{abstract}
   We consider a family of gradient Gaussian vector fields  on $\Z^d$, where the covariance operator is not translation invariant. 
   A uniform finite range decomposition of the corresponding covariance operators is proven, \ie the covariance operator can be written as a sum of covariance operators whose kernels are supported within cubes  of  increasing diameter. 
   An optimal regularity  bound for the subcovariance operators is proven.  We also obtain regularity bounds as we vary the coefficients defining the gradient Gaussian measures. 
   This extends a result of S. Adams, R. Koteck\'y and S. M\"uller. 
\end{abstract}

\begin{acknowledgements}
   The present results were obtained during my PhD studies.  I would like to express my gratitude to my advisor Prof. Stefan  M\"uller for the support and helpful discussions on the topic. I am also grateful to the  Bonn International Graduate school in Mathematics, Hausdorff Center for Mathematics, SFB 1060 and the Institute for Applied Mathematics in Bonn for the support and the nice environment. 
\end{acknowledgements}

\section{Introduction} 
\label{sec:intro-frd}

Recently, there has been some interest in the finite range decompositions of gradient Gaussian  fields on $\Z^{d}$.   
In particular, in \cite{MR2995704}, S.~Adams, R.~Koteck\'y and S.~M\"uller construct a  finite range decomposition for a family of translation invariant gradient Gaussian fields on $\Z^d $ ($d \geq 2$) which depends real-analytically on the quadratic from that defines the Gaussian field: they consider a large torus $\T^{d}_{N}:=(\Z/L^N \Z)^d$ and obtain a finite range decomposition with estimates  that do not depend on  $N$.

More precisely they consider a constant coefficient discrete elliptic system $\Acal = \nabla* A \nabla$ and show that its Green's function  $G(\cdot,\cdot) $ can be decomposed  as
\begin{equation*} 
   \begin{split}
      G_{A}(x,y) = \sum_{k} G_{A,k}(x,y)
   \end{split}
\end{equation*} 
where $G_{A}(\cdot,\cdot)$ have finite range \ie
\begin{equation*} 
   \begin{split}
      G_{A,k}(x,y) = 0 \qquad \text{whenever  } |x-y | > L^{k}
   \end{split}
\end{equation*} 
and they are positive definite \ie $\sum_{x,y} \varphi(x)G_{A,k}(x,y)\varphi(y) \geq 0$ for every $\varphi:\T^{d}_{N}\to \R^{m} $. 
Moreover they prove optimal estimates for $D^{\beta}\nabla^{\alpha }G_{A,k}$. 

We improve their result by extending it to the space dependent case. Namely, we consider an elliptic operator of the form $\Acal = \nabla* A \nabla$, where $A=A(x)$ is dependent on the space variable. Then we show that its Green's function can be written as the sum of positive and finite range functions $G_{A,k}(x,y)$

Looking at their proof this extension is highly non-trivial. Indeed, their proof uses both careful Fourier Analysis and Combinatorial techniques, which due to the space dependence, neither of them  seem to apply. Our approach takes a different route: we use $L^{p} $-theory arguments. Because some of this well-known $L^p$-estimates are not present in the discrete setting, we also need to prove the $L^p$ estimates for the discrete setting. 
As a byproduct, we are also able to prove the equivalent of the Finite range Decomposition in the continuous setting which to our knowledge is also not known.

The manuscript is organized as follows: in \Sec~\ref{sec:Preliminary Results}, we give a brief introduction to the results contained in \cite{MR2995704}, introduce some notation; in \Sec~\ref{sec:hypothesis} state our main result; in \Sec~\ref{sec:outlie} we give an outline of the proof in the continuous setting, hoping that this will make the proof easier to understand due to smaller notation, in \Sec~\ref{sec:construction-frd} we briefly discuss the construction of the finite range decomposition;  in \Sec~\ref{sec:discrete-estimates} we show extend $L^p$-theory to the discrete setting and show how to obtain the bounds; finally in \Sec~\ref{sec:analytic-dependence} we briefly discuss how to prove the bounds the derivative of $A$. Because the construction and the analyticity (\Sec~\ref{sec:construction-frd}, \Sec~\ref{sec:analytic-dependence}) are basically the same as in \cite{MR2995704}, we only sketch their proof.

\section{Preliminary Results} 
\label{sec:Preliminary Results}

In this section we are going to describe \emph{briefly} the results in \cite{MR2995704}.

Before writing precisely the statements contained in \cite{MR2995704}. We would like to introduce some notation. 
We will fix  a positive integer $N$ and odd integer $L>3$. 
The torus of size $N$ is defined as $\T^d_{N}:= (\Z/L^{N}\Z)^d $. The space of all function on $\T^{d}_N$ with values in $\R^m $ will be denoted by 
\begin{equation*} 
   \begin{split}
      \mathbf{X}_{N}:= (\R^{m})^{\T^{d}_{N}}= \insieme{\varphi:\Z^{d}\to \R^{m}:\ \varphi(x+z) = \varphi(x),\ \forall \varphi (L^{N}\Z)^d}.
   \end{split}
\end{equation*} 

This space will be endowed with with $\ell_2 $-scalar product, \ie
\begin{equation*} 
   \begin{split}
      \scalare{\varphi,\psi} =\sum_{x\in \T^{d}_{N}} \scalare{\varphi(x),\psi(x)}_{\R^{m}}.
   \end{split}
\end{equation*} 

In the last section, the $\mathbf{X}_{N} $ will be complexified and  will be substituted by the appropriate Hermitian inner product.  

We also define   
\begin{equation*} 
   \begin{split}
      \dist(x,y)&:=\inf\insieme{\abs{x-y+z}\colon z\in (L^N\mathbb Z)^d},\\
      \dist_\infty(x,y)&:= \inf\insieme{\abs{x-y+z}_\infty\colon  z\in (L^N\mathbb Z)^d},
   \end{split}
\end{equation*} 
and with a slight abuse of notation
\begin{equation*} 
   \begin{split}
      \dist_\infty(x,M) := \min \{ \rho_\infty(x,y)\colon y \in M \}.
   \end{split}
\end{equation*}

Gradient Gaussian fields are naturally defined on 
\begin{equation}
   \mathcal{X}_N:=\{\varphi\in \mathbf{X}_N: \sum_{x\in \mathbb T_N }\varphi(x)=0\}.
\end{equation}

For any set $M \subset \Lambda_N$, we define its closure by 
\begin{equation}
   \overline M=\{x\in \Lambda_N\colon \dist_\infty(x,M)\le 1\}. 
\end{equation}

The forward  and backward derivative are defined as
\begin{equation}
      (\nabla_j \varphi)(x):=\varphi(x+ e_j)-\varphi(x)\quad \text{and}\quad (\nabla^*_j\varphi)(x):=\varphi(x-e_j)-\varphi(x), 
\end{equation}

Until the end of this section we will denote by $A:\R^{m\times d}\to\R^{m\times d}$ a linear, symmetric and positive definite matrix.

The Dirichlet form on $\mathcal{X}_N$ is defined by, 
\begin{equation*}
   \scalare{\varphi,\psi}_+:=\sum_{x\in \T^{d}_N} \scalare{ A(\nabla\varphi(x)),\nabla\psi(x)}_{\R^{m\times d}},
\end{equation*}
where $\varphi,\psi: \mathcal{X}_{N}\to \R^{m}$. 

It is not difficult to notice that $(\cdot,\cdot)_+$, defines a norm on $\mathcal{X}$. 
Moreover, we will use $\|\cdot \|_{2}$ and $\|\cdot \|_{-} $ to denote the standard $\ell_{2}$ and the dual norm of $\|\cdot \|_{+}$; we will use $\cH_+ ,\  \cH ,\  \cH_-$ to denote $\mathcal{X}$ endowed with the norms $\|\cdot \|_{+} $, $\|\cdot \|_{2} $ and $\|\cdot \|_{-} $ respectively. 

Consider now the Green's operator  $\Cs_A:={\As}^{-1}$ of the operator $\As$ and the corresponding bilinear form on $\mathcal{X}_N $ defined by
\begin{equation*}
   \mathcal{G}_A(\varphi,\psi)=\langle\Cs_A\varphi,\psi\rangle=(\varphi,\psi)_-,\quad \varphi,\psi\in\mathcal{X}_N.
\end{equation*}

Given that the operator $\As$  and its inverse commutes with translations on $\T_N$, there exists a unique kernel  $\Ccal_A$ such that
\begin{equation}
   (\Cs_A\varphi)(x)=\sum_{y\in\T_N}\Ccal_A(x-y)\varphi(y).
\end{equation}
It is easy to see that  the function $G_{A,y}(\cdot)=\Ccal_A(\cdot-y)$ is the unique solution(with zero-mean) of the equation 
\begin{equation}
   \As G_{A,y}=\bigl(\delta_y -\frac1{L^{Nd}}\bigr) \I,
\end{equation}
where $\I$ is the unit $m\times m$ matrix.

Notice that for any $a\in\R^m$ one has:
\begin{equation*} 
   \begin{split}
      (\As G_{A,y}) =\bigl(\delta_y -\frac1{L^{Nd}}\bigr)  \in \mathcal{X}_N.
   \end{split}
\end{equation*}

In \cite{MR2995704}, among other things, the following result is proved:  

\begin{theorem}[{\cite{MR2995704}}]
   \label{thm:AKM_FRDfamily} 
   For any $d \geq2  $ and any multiindex $\alpha $, there exists a constant $C_{\alpha}(d),\ \eta_{d}(\alpha)$ such that the following properties hold:

For any integer $N\geq1 $, every $k=1,\ldots,N+1$ and every odd integer $L\geq 16 $, the map $A\mapsto C_{A,k}$ is real-analytic and
\begin{enumerate} 
   \item There exist positive definite operators $\mathcal C_{A,k}$ such that 
      \begin{equation*} 
         \begin{split}
           \mathscr C_{A} = \sum_{k=1}^{N+1} \mathscr C_{A,k}.
         \end{split}
      \end{equation*} 
   \item There exist constants $C_{A,k} $ such that 
      \begin{equation*} 
         \begin{split}
            \mathcal{C}_{A,k} = C_{A,k} \quad\text{whenever }\dist_{\infty}(x,0) > 1/2L^{k}
         \end{split}
      \end{equation*} 
   \item Let $A_{0} $ be such that $\scalare{A_{0}F,F}_{\R^{m}} \geq c_{0}\|F\|_{\R^{m\times d}}$.  Then
      \begin{equation*} 
         \begin{split}
            \sup_{\|\dot{A}\|\leq 1} \big\| (\nabla^{\alpha} D^{j}_{A} \mathcal C _{A_{0},k}(x)(\dot{A},\ldots,\dot{A})) \big\|\leq C_{\alpha}(d) \big ( \frac{2}{c_{0}}\big)^{j} j! L^{-(k-1)(d-2+|\alpha|)}L^{\eta_d(\alpha)},
         \end{split}
      \end{equation*} 
      where $D^{j}_A $ denotes the $j $-th derivative with respect to $A $ and $\|A\|$, denotes the operator norm of a linear mapping $A:\R^{m\times d}\to \R^{m\times d}$. 
\end{enumerate} 
\end{theorem}

\section{Notation and Hypothesis} 
\label{sec:hypothesis}

Let $\bar A:\T^d\to \Lcal_{\rm sym}(\R^{m\times d})$ be a $C^{3}$ function, where 
$\Lcal_{\rm sym}(\R^{m\times d})$ is the space of linear maps on $\R^{m\times
   d}$ such that $A=A^{*}$ and the associated operator is elliptic, namely there
exists a constant $c_1,c_0 >0$ such that 
\begin{equation}
   \label{eq:ellipticity}
   \begin{split}
      c_1 |P|^2 \geq \bar A_{i,j}^{\alpha,\beta} P_{\alpha}^{i} P_\beta^j \geq 
      c_0 |P|^2\qquad \forall P\in \R^{m\times d}
   \end{split}
\end{equation}
and  there exists an $\varepsilon_0>0$ (small enough)  
such that 
\begin{equation}
   \label{eq:cond-bar-A}
   \begin{split}
      \sum_{|\gamma  |\leq 3} \sup_{\T^d} |D^{\gamma }  \bar{A}_{i,j}^{\alpha,\beta}| \leq \varepsilon_{0},
   \end{split}
\end{equation}
where $\gamma $ is a multi-index.

For every $N>1$, we define the function $A_N:\T^d_N\to \Lcal_{\rm sym}(\R^{m\times d})$  in the following natural way:
\begin{equation}
   \label{eq:def-A_N}
   \begin{split}
      A_{N}(x)=\bar A(x/L^{N}).
   \end{split}
\end{equation}

The condition~\eqref{eq:cond-bar-A}, can be expressed in terms of $A_{N}$ as
\begin{equation} 
   \label{eq:cond-A_N}
   \begin{split}
      \sup_{|\gamma  |\leq 3}\sup_{\T^d_{N}} L^{N|\gamma  |}|\nabla 
      ^{\gamma } ( {A}_{N} )_{i,j}^{\alpha,\beta}| \leq \varepsilon_{0}.
   \end{split}
\end{equation} 
On the other hand, if there exists a $A_{N}$ such that \eqref{eq:cond-A_N} holds, then by some elementary interpolation one can construct a $\bar{A}$ such that \eqref{eq:def-A_N}  holds.  

Given that we will mainly work for $N$ fixed, if it is clear from the context we will drop the $N$-subscript.

We denote by $\Ecal\subset \insieme{q:\ \T^{d}_{N}\to \Lcal_{\rm sym}(\R^{m\times d})}$ such that 
there exist  constants $c_{0},c_{1} \geq 0 $ such that for every 
$x\in T^{d}_{N} $  and $F\in M_{\rm sym}(\R^{m\times d})$, it holds 
\begin{equation*} 
   \begin{split}
      c_{0}  \scal{F}{F} \leq\scal{q(x) F}{F}\leq c_{1} \scal{F}{F}.
   \end{split}
\end{equation*} 
The space $\Ecal$, is not a vector space. 
It will be endowed with the distance induced by the norm norm
\begin{equation*} 
   \begin{split}
      \|q \|_{\Ecal} =\sup_{x\in\T^{d},|\beta|\leq 3} \| 
      L ^{|\beta|N}\nabla^{\beta} q(x)\|_{M_{\rm sym}( \R^{m\times d} )},
   \end{split}
\end{equation*} 
where $\beta $ is a multiindex.  

Similarly as before, we introduce the following notations:
\begin{equation}
   \mathcal{X}_N:=\{\varphi\in \mathbf{X}_N: \sum_{x\in \mathbb T_N }\varphi(x)=0\},
\end{equation}
and 
\begin{equation*}
   \begin{split}
      \As\colon \cH_+\to\cH_-,\quad \varphi\mapsto \As\varphi:=\nabla^*(A\nabla\varphi).
   \end{split}
\end{equation*}
As in \Sec~\ref{sec:intro-frd}, let 
$\Ccal_{A}:\T^{d}_{N}\times  \T^{d}_{N} \to \R^{m\times d}$ such that
\begin{equation*}
   \As \Ccal_{A,y}=\bigl(\delta_y -\frac1{L^{Nd}}\bigr).
\end{equation*}

We will extend Theorem~\ref{thm:AKM_FRDfamily} in the following way:

\begin{theorem}\label{thm:FRD-mio}
   Let $d\geq 3$,  $A_{N}$ be defined as above.   
   Then there exists $\varepsilon _{0}>0$, $C_{d}(\alpha)$ and  $\eta_{d}(\alpha)$, such that for every $\varepsilon <\varepsilon _0$
   the operator $\Cs_A \colon \cH_-\to\cH_+$, where $\|A\|_{\Ecal}\leq \varepsilon$, admits a finite range 
   decomposition, i.e., there exist  positive-definite 
   operators 
   \begin{equation}
      \Cs_{A,k} \colon \cH_-\to\cH_+,\  (\Cs_{A,k}\varphi)(x)=\sum_{y\in\T^{d}_N}\Ccal_{A,k}(x,y)\varphi(y),\     k=1,\dots, N+1,
   \end{equation}
   such that 
   \begin{equation*}
      \Cs_A=\sum_{k=1}^{N+1} \Cs_{A,k},
   \end{equation*}
   and for associated kernel $\Ccal_{A,k}$, there exists a constant matrix $C_{A,k}$ such that
   \begin{equation*}
      \Ccal_{A,k}(x,y)=  C_{A,k}\ \text{ whenever } \   \dist_\infty(x,y)\geq \frac{1}{2} L^k\quad \mbox{ for } k=1,\dots,N  .
   \end{equation*}
   Moreover,  if 
   $(A_0 F, F)_{\R^{m\times d}} \geq  c_0 \norm{F}_{\R^{m\times d}}^2$  for all 
   $F \in \R^{m \times d}$ and $c_0 > 0$ and if 
   $\|A\|_{\Ecal} \leq 1/2$ then 
   \begin{equation*}
      \sup_{\norm{\dot{A}}\le 1}\Big\|\big(\nabla^{\alpha}_{y}D_A^j\Ccal_{A_0,k}(x,y)(\dot{A},\ldots,\dot{A})\Big\|
      \le C_{\alpha}(d) \left(\frac{2}{c_0}\right)^j  j! \, L^{-(k-1)(d-2+|\alpha|)}L^{\eta(\alpha, d)}.
   \end{equation*}
\end{theorem}

\section{Outline of the proof in the continuous case} 
\label{sec:outlie}

Before going to the discrete setting, we would like to briefly expose the basic idea 
in the continuous case. 

In what follows, we will use the symbol $\lesssim $ to indicate an inequality is 
valid up to universal constants depending eventually on the dimensions $d,m$.  

For the sake of simplicity, we take $A=A(x)$ be elliptic with $A$ smooth.

Let $B$ be a ball,  $\Pi_B: W^{1,2}(\R^n)\to W^{1,2}_{0} (B)$ be the 
projection operator.   Moreover, we define $P_B := \id -\Pi_{B}$.

The construction technique is due to Brydges \etal (see 
\cite{MR2070102,MR2240180}) and consists in  considering  the operators 
\begin{equation*}
   \begin{split}
      \Ts_{B} f := \frac{1}{|B|}\int _{\T^d} \Pi_{x+B} f\dx \qquad \text{and} \qquad
      {\Rs_B} := \id - \Ts_B. 
   \end{split}
\end{equation*}
Let $r_1,\ldots,r_k>0$ and $B_{r_1},\ldots,B_{r_k}$ be the balls of radius $r_k$ centered in $0$. 
Whenever it is clear from the context, we will denote by $\Rs_{k}:=\Rs_{B_{k}}$.

The operators $\Cs_k$ that appear in the Theorem~\ref{thm:AKM_FRDfamily} and Theorem~\ref{thm:FRD-mio}, will be of the form
\begin{equation*}
   \Cs_k : =(\Rs_{1}\dots \Rs_{k-1} )\Cs(\Rs_{k-1}^\prime\dots \Rs_{1}^\prime)
   - (\Rs_{1}\dots \Rs_{k-1} \Rs_k)\Cs(\Rs_k^\prime  \Rs_{k-1}^\prime  
   \dots \Rs_{1}^\prime), \ k=1,\dots, N,  
\end{equation*}
for a particular choice of $\{ r_{k}\}$. 

Then the proof of the finite range property will follow by abstract reasoning (see $\S$~\ref{sec:construction-frd}).

In \cite{MR1354111}, among other things the authors show:
\begin{thm}[{\cite[Theorem~1]{MR1354111}}] 
   \label{thm:dolzman-mueller-orig-thm1}
   Let $\Omega $ be  a regular domain and 
   $A^{\alpha ,\beta }_{i,j}\in C^{k,\alpha }(\bar{\Omega })$ for some $\alpha \in (0,1)$ such that
   \begin{equation*} 
      \begin{split}
         A_{i,j}^{\alpha ,\beta } P^{i}_{\alpha } P^{j}_{\beta } > c |P 
         |^{2}, \qquad \text{for some } c>0 \text { and every } P\in \R^{d\times m}.
      \end{split}
   \end{equation*} 
   Then there exists a matrix $G_{y}$ such that
   \begin{equation*} 
      \begin{split}
         -D_{\alpha }(A_{i,j}^{\alpha ,\beta }D_{\beta }( G_{y} )^{j}_{k}) = 
         \delta _{i,k}\delta _{j}\qquad\text{in  } \Omega
      \end{split}
   \end{equation*} 
   in the sense of distributions and
   \begin{equation*} 
      \begin{split}
         G_{y}=0\qquad\text{ on }\partial \Omega.
      \end{split}
   \end{equation*} 
   Moreover, it holds
   \begin{equation*} 
      \begin{split}
         |D^{\nu } G(x,\cdot ) | \leq C |x-y |^{2-d - |\nu  |},
      \end{split}
   \end{equation*} 
   where $\nu $ is  a multi-index such that $|\nu  |\leq k$.  

\end{thm}

The above theorem is proven by using the following well-known $L^{p}$-estimates.

\begin{lemma} 
   \label{lemma:dolz-mue-orig-2}
   Suppose the same hypothesis as in Theorem~\ref{thm:dolzman-mueller-orig-thm1} 
   and let $p\in (1,\infty )$, $q\in (1,n)$.  
   \begin{enumerate} 
      \item If 
         $f\in L^{p}(\Omega,\R^{m\times d} ), F\in  L^{q}(\Omega ,\R^{m})$, then the system   
         \begin{equation*} 
            \begin{split}
               - D_{\alpha }(A_{i,j}^{\alpha ,\beta } D_{\beta }u^{j}) = 
               D_{\alpha } f^{\alpha }_{j} + F^{i} \qquad \text{in }\Omega,
            \end{split}
         \end{equation*} 
         with boundary condition 
         \begin{equation*} 
            \begin{split}
               u=0\qquad \text{on } \partial \Omega,
            \end{split}
         \end{equation*} 
         has a weak solution in $W^{1,s}(\Omega ; \R^{m})$, where
         \begin{equation*} 
            \begin{split}
               s=\min(p,q^{*}), \qquad q^{*}=\frac{nq}{n-q},
            \end{split}
         \end{equation*} 
         and 
         \begin{equation*} 
            \begin{split}
               \|u \|_{W^{1,s}} \leq C(\|f \|_{L^{p}} + \|F \|_{L^{q}}).
            \end{split}
         \end{equation*} 

      \item If $f\in L^{p,\infty },\ F\in L^{q,\infty }$ then there exists a 
         weak solution that satisfies
         \begin{equation} 
            \begin{split}
               \|u \|_{L^{s^{*},\infty }} + \|Du \|_{L^{s,\infty }}\leq C(\|f \|_{L^{p, \infty }} + \|Du \|_{L^{q,\infty }}).
            \end{split}
         \end{equation} 
   \end{enumerate} 
\end{lemma}

To simplify the notation we will write $\diver(A \nabla u)$ instead of $D_{\alpha }(A_{\alpha ,\beta}^{i,j}D_\beta u^j)$.

\begin{lemma} 
   \label{lemma:dolz-mue-orig-3}
   Suppose the same hypothesis as in Theorem~\ref{thm:dolzman-mueller-orig-thm1}. 
   Let $B_{2r}$ be a ball of radius $2r$ centered in $0$, $p>d$ and let $u$ be a 
   solution to  
   \begin{equation*} 
      \begin{split}
         \diver (A \nabla u)=0 \qquad \text{in }B_{2r}.
      \end{split}
   \end{equation*} 
   Then 
   \begin{equation*} 
      \begin{split}
         \sup_{B_{r}} |u |\leq r^{- n/q} M + r^{1- n/p} \|f \|_{B_{2r}},
      \end{split}
   \end{equation*} 
   where
   \begin{equation*} 
      \begin{split}
         M=\|Du \|_{L^{q,\infty }(B_{2r})} + \|u \|_{L^{q^{*},\infty }(B_{2r})}.
      \end{split}
   \end{equation*} 

\end{lemma}

\begin{proposition}
   \label{proposition:pre_kryesorja}
   Let $B_1,\dots,B_k$ be balls with radii $r_1,\cdots,r_k$ respectively.  Then, there exists a  dimensional constant $C_d$, such that
   \begin{equation*}
      \begin{split}
         \sup|\nabla ^{j }u|\leq C_{d}^{k}   \max\left( |x-y|,\dist(y, {B}_{1}^{C}),\ldots,\dist(y, B_k^{C})  \right)^{2-d+j},
      \end{split}
   \end{equation*}
   where $u= (P_{B_1}\cdots P_{B_k} C(x,\cdot))$ and $C(x,y)$ is the Green's function and $j<d-2$. 
\end{proposition}

\begin{proof}
   Let us sketch the proof of the above fact. 
   In the discrete case it will be done in more detail. 

   The proof will follow by induction.   

   Let $B_{1}$ be a  ball in generic position of size $r_1$.   
   Given that $\diver(A\nabla C_{x}(y))=0$, if $x\not\in B_1$ then $\Pi_{B_1}C(x,y)=0$, thus $P_{B_1}C(x,y)=C(x,y)$, hence the inequality follows from Theorem~\ref{thm:dolzman-mueller-orig-thm1}.

   Let  $\varepsilon:=\dist(y, {B}^{C}_1) <r_1$.   
   If $|x-y|>\varepsilon/2$, then by estimating the different terms 
   $\Pi_{B_1}C(x,y)$ and $C(x,y)$ separately one has the desired result.  
   Indeed, $C(x,y)\lesssim |x-y |^{2-d}$.  Then by using an appropriate 
   version of Lemma~\ref{lemma:dolz-mue-orig-3} one has that 
   \begin{equation*} 
      \begin{split}
         |\Pi _{B_{1}} C(x,y)|\lesssim |x-y |^{2-d} M,
      \end{split}
   \end{equation*} 
   where
   \begin{equation*} 
      \begin{split}
         M=\|D \Pi _{B_1} C_{x} \|_{L^{d/d-2,\infty }(B_{1})} + \|\Pi _{B_1} C_{x}  \|_{L^{d/d-1,\infty }(B_{1})}. 
      \end{split}
   \end{equation*} 
   Then by using Lemma~\ref{lemma:dolz-mue-orig-2} one has that 
   \begin{equation*} 
      \begin{split}
         \|D\Pi _{B_{1}} C_x \|_{L^{d/(d-2),\infty}} + 
         \|\Pi _{B_{1}} C_x \|_{L^{d/(d-1),\infty}} \lesssim
         \|D C_x\|_{L^{d/(d-2),\infty}} + 
         \|C_x\|_{L^{d/(d-1),\infty}}< \tilde{C}_{d},
      \end{split}
   \end{equation*} 
   where $\tilde{C}_{d}$ is a constant depending only on the dimension $d$.

   The inductive step is done  in  a very similar way and the  higher derivative estimates follow similarly.   
\end{proof}

Let $B_1,\dots,B_k$ be $k$  balls centered in $0$, with radii $r_1,\dots,r_k$ respectively and let $C(\cdot,\cdot)$ be the Green's function.   We will denote by $C_k(x,\cdot ):= \Rs_k \cdots \Rs_1 C(x,\cdot)$.

Let us now give a simple calculation that will be useful in Theorem~\ref{thm:kryesorja}.

\begin{lemma}
   \label{lemma:2}
   Let $j> 1$ be an integer.  Then 
   \begin{equation*}
      \begin{split}
         \frac{1}{r^{d}}\int_{0}^{r} \max(\alpha, |r-\rho|)^{-j}\rho^{d-1}d\rho \lesssim \frac{\alpha^{1-j}}{r}.
      \end{split}
   \end{equation*}
   Indeed, let us denote by $I$ the right hand side of the previous equation.  With a change of variables one has 
   \begin{equation*}
      \begin{split}
         I&= \frac{1}{r^{d}}\int_{0}^{r-\alpha} |r-\rho|^{-j}\rho^{d-1}d\rho + 
         \int_{r-\alpha}^{r} \alpha^{-j} \rho^{d-1} \d\rho 
         \\&  =\frac{1}{r^{j}}\int^{1-\frac{\alpha}{r}}_{0} |1-t|^{-j} t^{d-1}\dt + 
         \int ^{1}_{1-\frac{\alpha}{r}} \alpha^{-j} t^{d-1}\dt 
         \\& =\frac{1}{r^{j}}\int^{1-\frac{\alpha}{r}}_{0} |1-t|^{-j} \dt + 
         \int ^{1}_{1-\frac{\alpha}{r}} \alpha^{-j} \dt  \leq 
         r^{-j} \left( \frac{\alpha^{1-j}}{r^{1-j}} - 1 \right ) + 
         \frac{\alpha^{1-j}}{r} \\ &\leq 
         \frac{2\alpha^{1-j}}{r}.
      \end{split}
   \end{equation*}
   If $j=1$, then 
   \begin{equation*}
      \begin{split}
         I&= \frac{1}{r^{d}}\int_{0}^{r-\alpha} |r-\rho|^{-1}\rho^{d-1}d\rho + 
         \int_{r-\alpha}^{r} \alpha^{-1} \rho^{d-1} \d\rho 
         \\&  =\frac{1}{r^{1}}\int^{1-\frac{\alpha}{r}}_{0} |1-t|^{-1} t^{d-1}\dt + 
         \int ^{1}_{1-\frac{\alpha}{r}} \alpha^{-1} t^{d-1}\dt 
         \\& =\frac{1}{r^{1}}\int^{1-\frac{\alpha}{r}}_{0} |1-t|^{-1} \dt + 
         \int ^{1}_{1-\frac{\alpha}{r}} \alpha^{-1} \dt  \leq 
         \frac{1}{r} \Big( \big|\log\big(\frac{\alpha}{r}\big) \big| +1\Big).
      \end{split}
   \end{equation*}
\end{lemma}

\begin{theorem}
   \label{thm:kryesorja}
   Let $C_k,B_i,r_i$ as above and  such that $r_1<\dots<r_h<|x-y|<r_h+1<\dots<r_k$.  Then,
   \begin{enumerate}
      \item  if $k -h< d-2$, then it holds 
         \begin{equation*}
            \begin{split}
               |C_{k}(x,y)|&\lesssim \frac{1}{r_{h+1}\cdots r_k} |x-y|^{2-d+k - 
                  h}\prod_{i=h+1}^{k}\left( \Big|\log\left(\frac{|x-y|}{r_i}\right)\Big|+1\right)\\
               |\nabla ^{j}_{y}C_{k}(x,y)| &\lesssim   \frac{1}{r_{h+1}\cdots 
                  r_k} |x-y|^{2-d+k -j -h},
            \end{split}
         \end{equation*}
      \item  if $k-h\geq d-2$, it holds
         \begin{equation*}
            \begin{split}
               |C_{k}(x,y)|&\lesssim \frac{1}{ r_{k-d+3}\cdots  r_k} \left|\log( |x-y| )\right|\\
               |\nabla ^{j}_{y}C_{k}(x,y)| &\lesssim   \frac{1}{ r_{k-d +2 
                     -j}\cdots  r_k} \prod_{i=h+1+j}^{k}\left( \Big|\log\left(\frac{|x-y|}{r_i}\right)\Big|+1\right).\\
            \end{split}
         \end{equation*}
   \end{enumerate}
\end{theorem}

\begin{proof}
   We will prove only (i). The proof of (ii) is very similar.

   Let us initially consider the case $k=1$.  
   For simplicity we denote $\Pi_z:= \Pi_{B_1 +z}$.   With simple 
   computations, one has
   \begin{equation*}
      \label{eq:01071357524110}
      \begin{split}
         \sup \left|C_1(x,y)\right| \leq \frac{1}{|B|}\int_{B_1+y}\sup |(\id - \Pi_{z})C(x,\cdot)| + 
         \sup\left|\frac{1}{|B|} \int_{(y+B_1)^{C}}  \Pi_{z} C(x,\cdot)\dz\right|.
      \end{split}
   \end{equation*}
   Because of the fact that for every $t\in B_1+z$ the function $\Pi _{z}C_{x}$ 
   is harmonic and has null boundary condition, one has that  the second term 
   in the right hand side of \eqref{eq:01071357524110}  is   null. Hence  it is enough to 
   prove a bound only on the first  
   term. Given that for every $z\in y+B$ it holds $\dist(y,z+B_1)=r_1 - 
   |z-y|$. Then, by using Proposition~\ref{proposition:pre_kryesorja}, one has that
   \begin{equation*}
      \begin{split}
         \sup |(\id - \Pi_{z})C(x,\cdot)| \leq
         \begin{cases}
            (r_1 - |z-y|)^{2-d}& \text{ if\ \  $
               r_1 - |y-z|\geq |x-y|$} \\ 
            |x-y|^{2-d}& \text{otherwise}.
         \end{cases},
      \end{split}
   \end{equation*}
   Thus,
   \begin{equation*}
      \begin{split}
         \sup |C_1 (x,y)| &\lesssim \int_{0}^{r_1- |y-x|} 
         |r_1-\rho|^{2-d}\rho^{d-1}\d\rho + \int_{r_1 - 
            |x-y|}^{r_1}|x-y|^{2-d}\rho^{d-1} \d\rho \\ &  
         \lesssim \frac{|x-y|^{3-d}}{r_1} - r_{1}^{2-d} + 
         \frac{|x-y|^{3-d}}{r_1}\lesssim \frac{|x-y|^{3-d}}{r_1}.
      \end{split}
   \end{equation*}

   Let us now turn to the general case $k<d -2$, and let $B_1,\dots,B_k$ be balls of radii $r_1,\dots,r_k$ centered at the origin.   
   From Proposition~\ref{proposition:pre_kryesorja}, we have that 
   \begin{equation*}
      \begin{split}
         &\sup | P_{z_1+B_1}\cdots P_{z_k+B_k} C(x,\cdot)|\leq 
         \max\insieme{|x-y|,r_1 - |z_1-y|,\dots,r_k- |z_k-y|}^{2-d}\\ 
         &\leq  \max\insieme{|x-y|}^{2-d+k}\cdot\max\insieme{|x-y|,r_k- 
            |z_k-y|}^{-1}\cdots\max\insieme{|x-y|,r_k- |z_k-y|}^{-1}\\ &\hspace{10cm} =:g(z_1,\dots,z_k).
      \end{split}
   \end{equation*}
   Thus,
   \begin{equation*}
      \begin{split}
         \sup \Rcal_1\cdots \Rcal_k C (x,\cdot )\leq 
         \int_{B_1\times\cdots\times B_k} g(z_1,\ldots,z_k) \dz_1\cdots\dz_k. 
      \end{split}
   \end{equation*}
   From Lemma~\ref{lemma:2} we have that
   \begin{equation*}
      \begin{split}
         \int_{B_1\times\cdots\times B_k} g(z_1,\ldots,z_k) \dz_1\cdots\dz_k 
         \leq \frac{1}{r_1\cdots r_k}|x-y|^{2-d+k}\prod_{i} (|\log (|x-y|)| 
         +\log(r_i) +1),
      \end{split}
   \end{equation*}
   which proves the desired result.

\end{proof}

\begin{corollary}
   Suppose that $|x-y|>1$ and let $B_1,\ldots,B_k$ and such that $r_i=L^{i}$ 
   with $L>1$. Then there exists $\eta (j,d)$ such that  
   \begin{equation*}
      \begin{split}
         \nabla ^{j}C_k(x,y) \lesssim \frac{L^{\eta(j,d)}}{L^{k(d-2 -j)}}.
      \end{split}
   \end{equation*}
   Indeed, given that $\Rs'_{k} = \As \Rs_{k} \Cs$ one has that 
   \begin{equation*} 
      \begin{split}
         \Rs_{1}\cdots\Rs_{k} \Cs \Rs'_{k}\cdots\Rs'_{1} = \Rs_{1}\cdots\Rs_{k} \cdot  \Rs_{k}\cdots\Rs_{1} \Cs
      \end{split}
   \end{equation*} 
   hence by using Theorem~\ref{thm:kryesorja}, one has the desired result.

\end{corollary}

\section{Construction of the finite range decomposition} 
\label{sec:construction-frd}

In this section, we will briefly describe the construction of the finite range decomposition. 
Let us \emph{stress} that main idea in the construction of the finite decomposition goes back to Brydges \etal (\eg  \cite{MR2070102,MR2240180}). 
Because the construction is rather well-known and general, in this section we will briefly sketch how such construction can be made. 
There are different versions of the construction above mentioned  construction.  We have in mind in particular a very closely related construction that can be found in \cite{MR2995704}.

Let $Q$ be a cube of size $l$  and let us denote for simplicity of notation we will  use $\Pi_x:=\Pi_{Q+x}$.

For every  $\varphi \in \cH_+$, define
\begin{equation*}
   \Ts(\varphi ) :=\frac{1}{l^d}\sum_{x\in \T^{d}_N}\Pi_x \varphi. 
\end{equation*}

One also introduces ${\Ts}^{\prime}: \cH_- \to\cH_-$ be the dual of $\Ts$ \ie
\begin{equation}
   \langle {\Ts}^{\prime}\varphi,\psi\rangle=\langle \varphi,\Ts\psi\rangle,\quad \varphi\in\cH_-, \psi\in\cH_+.
\end{equation}
It is not difficult to notice that
\begin{equation}
   \label{eq:propertiesOfTT'}
   {\Ts}^{\prime}=\As{\Ts}{\As}^{-1},  
   \quad (\Ts' \varphi, \psi)_-  = (\varphi, \Ts' \psi)_-, 
   \quad \mbox{and   } (\Ts' \varphi, \varphi)_- = (\Ts \As^{-1} \varphi,  \As^{-1} \varphi)_+.
\end{equation}

In order to construct the finite range decomposition we will also need  ${\Rs}:=\id-{\Ts}$ and its dual ${\Rs}^{\prime}=\id-{\Ts}^{\prime}$.

Using \eqref{eq:propertiesOfTT'} one has that
\begin{equation*} 
   \begin{split}
      \Rs' = \As \Rs \As^{-1}
   \end{split}
\end{equation*} 

Given that $0 \leq \scalare{\Ts\varphi,\varphi} \leq \scalare{\varphi,\varphi}$, and \eqref{eq:propertiesOfTT'}, for every  $\varphi \neq 0$  one has that $(\Ts' \varphi, \varphi)_- > 0, \quad (\Rs' \varphi, \varphi)_- > 0$ and $(\Ts' \varphi, \Ts' \varphi)_- \leq (\Ts' \varphi, \varphi)_-$.

Moreover, given a bilinear form on $\mathcal{X}_{N}$, there exists a (unique) linear map such that
\begin{equation*} 
   \begin{split}
      B(\varphi,\psi) = \scalare{\Bs \varphi,\psi}.
   \end{split}
\end{equation*} 
The map $\Bs $ can be represented as kernel, namely there exist a map $\Bcal$ such that 
\begin{equation*} 
   \begin{split}
      (\Bs\psi)(x) = \sum_{\xi\in \T^{d}_{N}} \Bcal(x,y)\psi(y).
   \end{split}
\end{equation*} 

Indeed, for our case when all the functions live in a finite dimensional vector space,  this is a simple linear algebra exercise.

For every  $M_1, M_2 \subset \T_N$, we will define the distance
\begin{equation}
   \dist_\infty(M_1,M_2) := \min \{ \dist_\infty(x,y)\colon x \in M_1, y \in M_2 \}.
\end{equation}

Let us define $\Cs_1:=\Cs- \Rs \Cs{\Rs}^{\prime}$.  As we saw $\Cs $ is positive. The crucial step in proving the finite range decomposition is proving that $\Cs_1 $ is finiterange and also positive definite.  The proof is a minor modification of the original one. 

Finally the finite range decomposition can be construced by an iterated application of the above. 
Namely, let $(l_{j})$ be an increasing sequence. We will apply the above procedure  $ Q_j $ instead of $Q$.  Namely, set
\begin{equation}
   \Cs_k : =(\Rs_{1}\dots \Rs_{k-1} )\Cs(\Rs_{k-1}^\prime\dots \Rs_{1}^\prime)
   - (\Rs_{1}\dots \Rs_{k-1} \Rs_k)\Cs(\Rs_k^\prime  \Rs_{k-1}^\prime  \dots \Rs_{1}^\prime), \ k=1,\dots, N,
\end{equation}
and
\begin{equation}
   \Cs_{N+1}  :=(\Rs_{1} \dots \Rs_{N-1}\dots \Rs_N)\Cs(\Rs_N^\prime  \Rs_{N-1}^\prime  \dots \Rs_{1}^\prime).
\end{equation}
By doing this we have the desired finite range decomposition.

\section{Discrete gradient estimates and $L^p$-regularity for elliptic systems} 
\label{appSobolev}
\label{sec:discrete-estimates}

Let us now introduce some of the norms that will be used in the sequel. Let $ Q=[0,n]^d\cap\Z^d$, be a generic cube. For $ p>0 $ denote 
\begin{equation}
   \norm{f}_{p,Q}=\Big(\frac{1}{|Q|}\sum_{x\in Q_n}|f(x)|^p\Big)^{1/p},
\end{equation}
where $|Q|:=\#Q$.

To simplify notation, we will write  $\sum_Q f := \sum_{i\in Q} f(i)$ and 
$f_Q:=|Q|^{-1}\sum_{Q}f$.  

Additionally, let us define 

\begin{equation*}
   \begin{split}
      f^\#(x):=\sup_{Q\ni x}\ \frac{1}{|Q|}\sum_{Q }\big| f - 
      f_Q\big|\dx \qquad\text{and}\qquad \|f\|_{\bmo}:=\sup_{x\in \T^{d}_{N}} |f^{\#}(x)|.
   \end{split}
\end{equation*}
The Maximal Operator is defined by
\begin{equation*}
   \begin{split}
      \Mcal f(x):=\sup_{Q\ni x}\ \frac{1}{|Q|}\sum_{Q }| f |\dx 
   \end{split}
\end{equation*}

Moreover, let
\begin{equation*}
   \|f\|_{p,\infty}= \inf \insieme{\alpha:\  \frac{1}{\lambda 
      }|\insieme{f>\lambda}|^{1/p} \leq \alpha,\ \text{for all $\lambda>0$}}
\end{equation*}
and
\begin{equation*}
   \|f\|_{p,\infty,Q}= |Q|^{-1/p} \inf \insieme{\alpha:\  \frac{1}{\lambda 
      }|\insieme{f>\lambda}\cap Q|^{1/p} \leq \alpha,\ \text{for all $\lambda>0$}}.
\end{equation*}

We now state  a version of Sobolev inequality (see \cite{MR961019,AKM10}). 

\begin{proposition}\label{Sobolevpropi-iv-in-chapter}
   \label{Sobolevpropi-iv}
   For every $p\ge 1$ and $m,M\in N$ there exists a constant $C=C(p,M,m)$ such that:
   \begin{enumerate}
      \item[(i)] If $1\le p\le d$, $\frac{1}{p^*}=\frac{1}{p}-\frac{1}{d} $, and $ q\le p^*$,  $q<\infty $, then
         \begin{equation}
            n^{-\frac{d}{q}}\norm{f}_q\le Cn^{-\frac{d}{2}}\norm{f}_2+Cn^{1-\frac{d}{p}}\norm{\nabla f}_p.
         \end{equation}

      \item[(ii)] If $ p> d $, then
         \begin{equation}
            \big|f(x)-f(y)\big| \le Cn^{1-\frac{d}{p}}\norm{\nabla f}_p \qquad\mbox{ for all } x,y\in Q_n.
         \end{equation}

      \item[(iii)] If $ m\in\N$, $1\le p\le \frac{d}{m}$,  $\frac{1}{p_m}=\frac{1}{p}-\frac{m}{d}$, and  $q\le p_m$, $q<\infty $, then
         \begin{equation}
            n^{-\frac{d}{q}}\norm{f}_q\le Cn^{-\frac{d}{2}} \sum_{k=0}^{M-1}\norm{(n\nabla)^kf}_2 +Cn^{-\frac{d}{p}}\norm{(n\nabla)^Mf}_p.
         \end{equation}

      \item[(iv)] If $ M=\lfloor\frac{d+2}2\rfloor$,  the integer value of $\frac{d+2}2$, then
         \begin{equation}
            \max_{x\in Q_n}|f(x)|\le Cn^{-\frac{d}{2}}\sum_{k=0}^M\norm{(n\nabla)^kf}_2.
         \end{equation}
   \end{enumerate}

\end{proposition}

\begin{lemma}[Caccioppoli inequality]
   Let $v$ be  such that $\diver (A \nabla v)=0$ for every $x\in Q_{M}$ then 
   \begin{equation*}
      \begin{split}
         \sum_{Q_m} |\nabla v(x)|^{2}\leq \frac{c_0^{4}}{( M-m )^{2}}\sum_{Q_M}|v-\lambda|^{2},
      \end{split}
   \end{equation*}
   where $c_0$ is the constant defined in \eqref{eq:ellipticity}.
\end{lemma}

\begin{proof}

   Let $0\leq \eta\leq1$ be a that $|\nabla \eta|\leq \frac{1}{M-m}$ and such 
   that $\eta\equiv 1$ on $Q_m$ and $\eta = 0 $ on $\T^d_{N}\setminus \bar Q_M$.   
   Then 
   \begin{equation*}
      \begin{split}
         \sum_{Q_M} (A \nabla u \cdot \nabla u )\eta^2   = \sum_{Q_M} A \nabla u \cdot  \nabla (\eta^2 (u- \lambda)) -\sum_{Q_M} A \nabla u \cdot  2 \eta ((u-\lambda) \otimes D\eta)    
      \end{split}
   \end{equation*}
   By hypothesis, the first term in the right hand side vanishes.  
   Using the previous formula and the ellipticity, one has that 
   \begin{align}
      \sum_{Q_M} \kabs{\nabla u}^2 \eta^2  & \leq c_0\sum_{Q_M} A \nabla u \cdot  2 \eta ((u-\lambda) \otimes D\eta)     \leq \frac{1}{2} \sum_{Q_M} \kabs{\nabla u}^2 \eta^2  + \frac{c_0^4}{2}  \sum_{Q_M} \kabs{D\eta}^2 \kabs{u - \lambda}^2 , 
   \end{align}
   from which one has that  
   \begin{equation*}
      \begin{split}
         \sum_{Q_{m}} \kabs{\nabla u}^{2}\leq \sum_{Q_{M}} \kabs{\nabla u}^{2}\eta^2\leq 
         \frac{c_0^{4}}{( M-m )^2} \sum_{Q_M} |u-\lambda|^2.
      \end{split}
   \end{equation*}

\end{proof}

\begin{lemma}[Decay estimates]
   \label{lemma:consequences_caccippoli}
   \label{lemma:decay_estimates}
   Let $v$ be such that $\diver (A \nabla v)=0$ on $Q_M$, with $M,  M/2\in \N$ and 
   $2m\leq M$.  Then,

   \begin{equation*}
      \begin{split}
         \sum_{Q_m} |u(x)|^{2}& \lesssim (m/M)^{d} \sum_{Q_M} |u(x)|^2,\\
         \sum_{Q_{m}} |u - ( u )_{m}|^{2}&\lesssim (m/M)^{d+2}|\sum_{Q_M}u - ( u )_{M}|^2.
      \end{split}
   \end{equation*}
\end{lemma}
\begin{proof}

   From the Caccioppoli's inequality, one has that
   \begin{equation*}
      \begin{split}
         \sum_{Q_{M/2}} |M \nabla u (x)|^2 \lesssim \sum_{Q_{M}} |u(x)|^{2}.
      \end{split}
   \end{equation*}
   Noticing that if $u$ is a solution then also $\nabla u$ is a solution, we have that 
   \begin{equation*}
      \begin{split}
         \sum_{Q_M} \|(M\nabla )^ju\|\lesssim \sum_{Q_{M}}|u(x)|^2,
      \end{split}
   \end{equation*}
   hence 
   \begin{equation*}
      \begin{split}
         M^{-d}\sum_{j=0}^{k}\sum_{Q_{M/2}} \|(M/2 \nabla )^{j}u\|\lesssim 
         M^{-d}\sum_{Q_M} \|u\|^2.
      \end{split}
   \end{equation*}
   Finally applying the Sobolev, inequality we have that 
   \begin{equation}
      \label{eq:04131397396068}
      \begin{split}
         \sum_{Q_m}\|u\|^2 \leq m^{d} \max_{Q_{M/2}} \|u\|^2 \leq (\frac{m}{M})^{d} \sum_{Q_M}\|u\|^2.
      \end{split}
   \end{equation}

   Let us now prove the second inequality.   Using the Poincar\'e inequality and than \eqref{eq:04131397396068}, we have that
   \begin{equation*}
      \begin{split}
         \sum_{Q_M}|u-(u)_m|^{2}&\leq m^2 \sum_{Q_m} |\nabla u|^{2}\lesssim m^2 
         \left(\frac{2m}{M}\right)^{d}\sum_{Q_{M/2}}|\nabla u|^{2}\\ &\lesssim \left(  
            \frac{m}{M}  \right)^{d+2} \sum_{Q_M} |u - (u)_{M}|^{2},
      \end{split}
   \end{equation*}
   where in the last step we have used the Caccioppoli inequality.  
\end{proof}

\begin{lemma}\label{dolzman-mueller-lemma-1}
   Let $p_1,p_2, q_1,q_2\in [1,\infty]$, $p_1\neq p_2$, $q_1\neq q_2$.  Let 
   $\theta \in (0,1)$ and define $p,q$ by
   \begin{equation}
      \label{eq:09131347499984}
      \frac{1}{p}=\frac{\theta}{p_1} + \frac{1-\theta}{p_2},\qquad \frac{1}{q}=\frac{\theta}{q_1} + \frac{1-\theta}{q_2}
   \end{equation}

   Suppose that $T$ is  a linear operator such that 
   \begin{equation*}
      \begin{split}
         \left( \frac{1}{|Q|} \sum_Q |T f|^{q_i}\right)^{\frac{1}{q_i}} \leq C_i\left( 
            \frac{1}{|Q|} \sum_Q | f|^{p_i}\right) ^{\frac{1}{p_i}}
      \end{split}
   \end{equation*}

   Then
   \begin{equation*}
      \begin{split}
         \|Tf\|_{q,\infty,Q} \leq C_3 \|f\|_{p,\infty,Q},
      \end{split}
   \end{equation*}
   where $C_3$  depends on $\theta$, $C_1$, $C_2$.  

\end{lemma}

\begin{proof}
   The proof  of this result is well-known (see \eg \cite[Theorem~3.3.1]{butzer_berens}).  
   For completeness, we report an adapted elementary proof from \cite[Lemma~1]{MR1354111}. Let $p_1 < p_2,$  $q_1 < q_2$ and 
   $p$ is as in \eqref{eq:09131347499984}.  Assume 
   that $\|Tf\|_{q_i}\leq C_{i} \| f\|_{p_i}$ with $i=1,2$.  Let $\gamma >0 $ define
   \begin{equation}
      f_1= \begin{cases}
         f & \qquad \text{if $|f|>\gamma$}\\
         0& \qquad \text{if $|f|\leq \gamma$}
      \end{cases}
   \end{equation}
   and
   \begin{equation}
      f_2= \begin{cases}
         0 & \qquad \text{if $|f|>\gamma$}\\
         f& \qquad \text{if $|f|\leq \gamma$}.
      \end{cases}
   \end{equation}
   Given that
   \begin{equation*}
      \begin{split}
         \frac{1}{|Q|} \sum_{Q} |f_1|^{p_1}   \leq  \frac{p_1}{p-p_1} \gamma^{p_1-p} \|f\|^{p}_{p,\infty,Q}
      \end{split}
   \end{equation*}

   we have that
   \begin{equation*}
      \begin{split}
         \Big|\insieme{|Tf_{1}|>\frac{\alpha}{2}}\Big|&\leq  A^{q_1}_{1} 
         \big(\frac{2}{\alpha} \big)^{q_1} \|f_1\|_{p_1}^{q_1}\\ & \leq
         A^{q_1}_1 \big( \frac{2}{\alpha}  \big )^{q_1} 
         \Big(\frac{p_1}{p-p_1}\Big)^{q_1/p_1}\gamma ^{q_1- 
            pq_1/p_1}\|f\|_{p,\infty,Q}^{pq_1/p_1}\\ &= B_1 \alpha^{-q_1}\gamma^{q_1-pq_1/p_1}
      \end{split}
   \end{equation*}
   and similarly 

   \begin{equation}
      \Big|\insieme{|Tf_2|\geq \frac{\alpha}{2}}\Big|\leq B_2\alpha^{- q_2} \gamma^{q_2-pq_2/p_2}.
   \end{equation}

   Now 

   \begin{equation*}
      \|Tf\|_{q,\infty}^{q}= \sup_{\alpha}\alpha^q |\insieme{|Tf|>\alpha}|
   \end{equation*}
   and now using the triangular inequality, we have 
   \begin{equation*}
      \begin{split}
         \alpha^q|\insieme{|Tf|>\alpha/2}|&\leq \alpha^q|\insieme{|Tf_1|>\alpha/2}| +\alpha^q|\insieme{|Tf_2|>\alpha/2}|
         \\ &\leq B_1 \alpha^{-q_1}\gamma^{q_1-pq_1/p_1}+ B_2\alpha^{- q_2} \gamma^{q_2-pq_2/p_2}.
      \end{split}
   \end{equation*}

   One can archive the desired result by choosing $\gamma=\alpha^\beta$ where 
   $\beta= \big( \frac{q}{q_1} -\frac{q}{q_2}\big) \big(\frac{p}{p_1} -\frac{p}{p_2}\big)^{-1}$.  

\end{proof}

\begin{theorem}[Marcinkiewicz interpolation theorem]
   Let $0 < p_0, p_1, q_0, q_1 \leq \infty$ and $0 < \theta < 1$ be such that 
   $q_0 \neq q_1$, and $p_i \leq q_i$ for $i=0,1$. Let $T$ be a sublinear 
   operator which is of weak type $(p_0,q_0)$ and of weak type $(p_1,q_1)$. 
   Then $T$ is of strong type $(p_\theta,q_\theta)$. 
\end{theorem}
\begin{proof}
   The proof is well-known.   
\end{proof}

\begin{remark}
   \label{rmk:norma-debole-K}
   Let $K:\T^{d}_{N}\times \T^{d}_{N}\to \R^{d\times m}$ be such that 
   $|K(x,y)|\leq |x-y|^{2-d}$. Then has that
   \begin{equation*} 
      \begin{split}
         \|K(x,\cdot)\|_{L^{\frac{n}{n-2}},\infty}\leq 1
         ,\qquad\text{and}\qquad    \|K(x,\cdot)\|_{L^{\frac{n}{n-2}},Q,\infty}\leq 1.
      \end{split}
   \end{equation*} 
   Indeed, fix $t>0$ then
   \begin{equation*}
      \begin{split}
         |\insieme{y:\ |K(x,y)| > t}|\leq |\insieme{y: \ |x-y |^{2-d} >t}| = 
         |\insieme{y:\ |x-y|< t^{-(2-d)}} |\leq t^{-\frac{d}{d-2}}.
      \end{split}
   \end{equation*}  
\end{remark}

Let us recall the celebrated Hardy-Littlewood maximal theorem:

\begin{thm}
   \label{thm:hardy-littlewood-maximal}
   Let $f:\T^{d}_{N}\to \R^m$.   Then
   \begin{equation*}
      \begin{split}
         |\Mcal f|_{p}\leq |f|_{p}
      \end{split}
   \end{equation*}
\end{thm}

\begin{theorem}[Fefferman-Stein]
   \label{thm:feffermanstein}
   Let $Q$ be a cube and let $f:Q\to \R^m$ such that $\sum_{Q} f = 0$. Then there 
   exists constants $C_1,C_2$ such that 
   \begin{equation*}
      \begin{split}
         \|\Mcal f\|_{p,Q}\leq C_{1} \|f^{\#}\|_{p,Q} \qquad \text{and} \qquad 
         \|f^{\#}\|_{p,Q}\leq C_2\|\Mcal f\|_{p,Q}.
      \end{split}
   \end{equation*}
\end{theorem}

\begin{proof} 
   The proof follows from the classical Fefferman\&Stein result after one does a piecewise linear interpolation of the function $f:Q\to \R^{m}$.

\end{proof}

\begin{corollary}
   Let $T$ be an linear operator such that for every $f:Q\to \R^{m}$. 
   Then for every $q > p$,  there exists a constant $C:=C(p)$ such that for every $f:Q\to \R^{m}$  it holds 
   \begin{equation*} 
      \begin{split}
         \sum_{x \in Q} |Tf^\#(x) |^{p}  \leq \sum_{x\in Q} |f(x)|^{p}.
      \end{split}
   \end{equation*} 
\end{corollary} 

\begin{proof}
   The map $f\mapsto (Tf)^\#$ is a sublinear  and a bounded  map from 
   $\Lp^{\infty}(\mathcal{X})\to \Lp^{\infty}(\mathcal{X})$ which is of weak type $(p,p)$ and of weak type $(\infty,\infty)$.
   Then for every $q\geq p$, it holds that $f\mapsto (Tf)^\#$ is bounded.
   This implies that $f\mapsto M(Tf)$ is bounded because Theorem~\ref{thm:feffermanstein} 
   and hence $f\mapsto Tf$ is bounded.   
\end{proof}

In the next lemma $A=A_0$ is a constant positive definite operator.   

Let us now recall a classical result.  We also provide a proof for completeness.  

\begin{lemma}[{\cite[Lemma V.3.1]{MR717034} }]
   \label{lemma_iteration1}
   Assume that $\phi(\rho)$ is a non-negative, real-valued, bounded function defined on an interval $[r,R] \subset \R^+$. Assume further that for all $r \leq \rho < \sigma \leq R$ we have
   \begin{equation*}
      \phi(\rho) \leq \big[ A_1 (\sigma - \rho)^{-\alpha_1} + A_2 (\sigma - \rho)^{- \alpha_2} + A_3 \big] + \vartheta \phi(\sigma) 
   \end{equation*}
   for some non-negative constants $A_1, A_2, A_3$, non-negative exponents $\alpha_1\geq \alpha_2$, and a parameter $\vartheta \in [0,1)$. Then we have
   \begin{equation*}
      \phi(r) \leq c(\alpha_1,\vartheta) \big[ A_1 (R - r)^{-\alpha_1} + A_2 (R - r)^{- \alpha_2} + A_3 \big] \,.
   \end{equation*}

   \begin{proof}
      We proceed by iteration and start by defining a sequence $(\rho_i)_{i \in \N_0}$ via
      \begin{equation*}
         \rho_i := r + (1-\lambda^i) (R - r)
      \end{equation*}
      for some $\lambda \in (0,1)$. This sequence is increasing, converging to $R$, and the difference of two subsequent members is given by
      \begin{equation*}
         \rho_{i} - \rho_{i-1} = (1 - \lambda) \lambda^{i-1} (R - r) \,.
      \end{equation*}
      Applying the assumption inductively with $\rho=\rho_i$, $\sigma = \rho_{i-1}$ and taking into account $\alpha_1 > \alpha_2$, we obtain
      \begin{align*}
         \phi(r) & \leq A_1 (1-\lambda)^{-\alpha_1} (R - r)^{-\alpha_1} + A_2 (1-\lambda)^{-\alpha_2} (R - r)^{-\alpha_2} + A_3 + \vartheta \phi(\rho_1) \\
         & \leq \vartheta^k \phi(\rho_k) + (1-\lambda)^{- \alpha_1} \sum_{i=0}^{k-1} \vartheta^{i} \lambda^{-i \alpha_1} \big[ A_1 (R - r)^{-\alpha_1} + A_2 (R - r)^{- \alpha_2} + A_3 \big]
      \end{align*}
      for every $k \in \N$. If we now choose $\lambda$ in dependency of $\vartheta$ and $\alpha_1$ such that $\vartheta \lambda^{-\alpha_1} < 1$, then the series on the right-hand side converges. Therefore, passing to the limit $k \to \infty$, we arrive at the conclusion with constant $c(\alpha_1,\vartheta)= (1-\lambda)^{- \alpha_1} (1- \vartheta \lambda^{-\alpha_1})^{-1}$.
   \end{proof}
\end{lemma}

\begin{lemma}
   Let $u$  be a solution to 
   \begin{equation}
      \label{eq:main}
      \begin{split}
         \begin{cases}
            \Acal_0 u = \diver f, &\text{in }Q_{M},\\ 
            u =0\, &\text{in } \T^{d}_{N}\setminus \bar Q_{M}.
         \end{cases}
      \end{split}
   \end{equation}
   The map $f\mapsto \nabla u$ is a continuous map from $\Lp^{\infty}\to \bmo$
\end{lemma}

\begin{proof}
   Let $m\leq[M/2]$ and let  $u_1$ be such that 
   \begin{equation*}
      \begin{split}
         \begin{cases}
            \diver (A \nabla u_1) = \diver f &\text{in } Q_{M}\\
            u_1 =0 &\text{in }\T^{d}_N\setminus \bar Q_M
         \end{cases}
      \end{split}
   \end{equation*}
   and $u_0=u-u_1$.   Notice that $\diver (A \nabla u_0)=0$ in $Q_M$.  
   We have
   \begin{equation*}
      \begin{split}
         \sum_{Q_M} |\nabla u_1|^2\lesssim\sum_{Q_M} A \nabla u_1\cdot \nabla u_1 \leq\sum_{Q_M} 
         f \nabla u_1 \leq |f|_{\infty}M^{d/2} \left( \sum_{Q_{M}} |\nabla u_1|^2\right)^{1/2}
      \end{split}
   \end{equation*}
   from which we have that 
   \begin{equation*}
      \begin{split}
         \sum_{Q_M} |\nabla u_1|^{2} \leq M^{d} |f|^{2}_{\infty}
      \end{split}
   \end{equation*}

   Given that from Lemma~\ref{lemma:consequences_caccippoli}  we have that 

   \begin{equation*}
      \begin{split}
         \sum_{Q_m}| \nabla u_0 -(\nabla u_0)_{m}|^{2} \lesssim \left( \frac{m}{M}   
         \right)^{d+2} \sum_{Q_M} |\nabla u_0 - (\nabla u_0)_{M}|^2
      \end{split}
   \end{equation*}
   it follows that
   \begin{equation*}
      \begin{split}
         \sum_{Q_m} |\nabla u - (\nabla u)_m|^{2} \leq \left( \frac{m}{M}\right)^{d+2} 
         \sum_{Q_M}|\nabla u -(\nabla u)_{M}|^2 + \sum_{Q_m} |\nabla u_1|^{2}\leq \left( 
            \frac{m}{M}\right) ^{d+2} + M^{d} |f|_{\infty}^2
      \end{split}
   \end{equation*}
   Finally using Lemma~\ref{lemma_iteration1} we have the desired result.

\end{proof}

From now on  $A=A(x)$, namely depends on the space.  

The next lemma is an adaption of \cite[Lemma~2]{MR1354111} to 
the discrete case.  The original proof is based on an argument
in \cite{morrey}. We will rather use an argument based on Theorem~\ref{thm:feffermanstein}.

In the continuous case, the analog version of the next lemma can be found in~\cite[Lemma~2]{MR1354111}.

\begin{lemma} [Global estimate]
   \label{lemma:dolzman-mueller-lemma-2} Let $p \in  (1, \infty)$ $q\in (1,n)$
   \begin{enumerate}
      \item If $f: \T^d_N\to R^{md}$, 
         $g:\T^d_N\to \R^m$ and let $u$ be the solution of
         \begin{equation*}
            \begin{cases}
               -\diver (A\nabla  u) = \diver f + g& \text {in } Q_M\\
               u = 0 &\text{in } \T^d_N \setminus\bar Q_M
            \end{cases}
         \end{equation*}
         Then if 
         \begin{equation*}
            \begin{split}
               s=\min(p,q^*), \qquad q^{*}= \frac{dq}{d-q}
            \end{split}
         \end{equation*}
         we have
         \begin{equation*}
            \begin{split}
               \left(\sum_{Q_{M}} | \nabla u|^{s}\right)^{1/s} \lesssim  \left(  
                  \sum_{Q_M} | f|^{p}\right)^{1/p} + \left( \sum_{Q_M} |M g|^{q} \right)^{1/q}
            \end{split}
         \end{equation*}
      \item and 
         \begin{equation*}
            \|u\|_{s^*,\infty} + \|\nabla u\|_{s,\infty}\leq C \left( 
               \|f\|_{p,\infty,Q_M} + |g|_{q,\infty,Q_M}\right)
         \end{equation*}
   \end{enumerate}
\end{lemma}

\begin{proof}
   Let $x_0$ be the center of the cube $Q_M$.  For simplicity of notation we 
   will denote by $A_0:=A(x_0)$.  
   With simple algebraic manipulations we have 
   \begin{equation*}
      \begin{split}
         \diver(A_0\nabla u)= \diver (f + (A_0-A)\nabla u )   
      \end{split}
   \end{equation*}

   Let $\eta$ such that $\eta\equiv 0$ in $\T^d_N\setminus \bar Q_M$. Then we have 
   \begin{equation*}
      \begin{split}
         \diver ( A_{0} \nabla (u\eta)) = \diver \left(  (A_0-A) \nabla (u\eta)   \right) + G +\diver F
      \end{split}
   \end{equation*}
   where $G= g\eta + fD\eta + A(x)\nabla u D\eta$ and $F= f\eta + A(x)u D\eta$.  

   Let $w$ be defined as 
   \begin{equation*}
      \begin{split}
         \begin{cases}
            \diver ( \nabla w ) = - G & \text{in } Q_M\\
            w=0             & \text{in }\T^d_{N}\setminus \bar Q_M
         \end{cases}
      \end{split}
   \end{equation*}
   Hence, from the constant coeficient case one has that
   \begin{equation*}
      \begin{split}
         \left(\sum_{Q_M}\|M \nabla w\|^{r^*} \right) ^{1/r^{*}} \lesssim \left(\sum_{Q_M} \|G\|^{r}\right)^{\frac{1}{r}}
      \end{split}
   \end{equation*}
   Denoting with $\tilde{F}=F +\nabla w$ we have that 
   \begin{equation*}
      \begin{split}
         \diver (A_{0}\nabla (u\eta))      =  \diver 
         \left(A-A_0)\nabla v\right)  + \diver \tilde{F}   \qquad 
         \text{in } Q_M.   
      \end{split}
   \end{equation*}

   We will now make a fixed point argument.  Fix $V$ and consider the 
   linear operator  
   $T:V\mapsto v$ where $v$ is the solution of 
   \begin{equation*}
      \begin{split}
         \diver (A_0 \nabla v) = \diver \left(A-A_0)\nabla V\right) + 
         \diver \tilde{F} 
      \end{split}
   \end{equation*}
   The operator $T$ is continuous, namely
   \begin{equation*} 
      \begin{split}
         \sum_{x\in Q_{M}} | \nabla T(V_1 -V_2) | ^{s}\leq c \sup_{x\in 
            Q_{M}} |A(x) -A(x_{0}) |^{s} \sum_{x\in Q_{M}} |\nabla 
         V_{1}(x) - \nabla V_{2}(x) |^{s} + c \sum_{x\in Q} |\tilde{F} 
         |^{s}
      \end{split}
   \end{equation*} 

   If 
   \begin{equation} 
      \label{eq:cond-10171382046429}
      \begin{split}
         \sup_{x\in Q_{M}} |A(x)-A_{0} |\leq \frac{1}{2}A(x_0)
      \end{split}
   \end{equation} 

   one can apply 
   the fixed point theorem and deduce that the solution coincides with 
   $u\eta $, and that
   \begin{equation*}
      \begin{split}
         \left(\sum_{Q_M}|(M \nabla)u|^{s}\right)   ^{1/s} \leq C \left( \sum_{Q_M}|\tilde{F}|^{s} \right)^{1/s}.
      \end{split}
   \end{equation*}
   Finally the condition \eqref{eq:cond-10171382046429} is ensured by  \eqref{eq:cond-A_N}.
\end{proof}

For the continuous version of the following lemma see \cite[Lemma~4]{MR1354111}
\begin{lemma}
   \label{lemma:dolzmann_mueller_lemma_3}
   Let $q\in (1,d)$ $p>d$.   Let 
   \begin{equation*}
      T= \|\nabla u\|_{\Lp^{q,\infty}(Q_{2M})} + \|u\|_{\Lp^{q^*,\infty}(Q_{2M})}.
   \end{equation*}

   Suppose that $u$ satisfies
   \begin{equation*}
      -\diver (A \nabla u )= \diver f \qquad \text{in } Q_{2M}
   \end{equation*}
   Then there exists $m_{0}:=m_{0}(p,q)$ such that if $M> m_{0}$ then 
   \begin{equation*}
      \sup_{Q_m} |u| \lesssim  M^{- \frac{d}{q}}T +M^{1-\frac{d}{p}} \|f\|_{\Lp^p},
   \end{equation*}
   where $m = \big[M/d \big]$
\end{lemma}

\begin{proof}

   Let $\delta\in \N$ such that $\delta\leq M$. Set 
   $\kappa = \lfloor\frac{M}{\delta}\rfloor$ and  let $\varphi$ be such that 
   $\varphi \equiv 1$ in $Q_M$, 
   $\varphi \equiv 0$ in $\T^d_N\setminus \bar Q _{M+\delta}$, and such that 
   $|\nabla \varphi|\leq \frac{1}{\delta}$.    
   Then for every $p_1>0$ one has that  
   \begin{equation*}
      \begin{split}
         \left(\frac{1}{|Q_M|}\sum_{Q_M}|\nabla u|^{p_1}\right)^{\frac{1}{p_1}} 
         \leq \left(\frac{|Q_{M+\delta}|}{|Q_{M}|}\right)^{1/p_1}\left(\frac{1}{|Q_{M+\delta}|}
            \sum_{Q_{M+\delta}} |\nabla (\varphi u)|^{p_1} \right)^{\frac{1}{p_1}}
      \end{split}
   \end{equation*}

   With simple calculations one has that 
   \begin{equation}
      \label{eq:01141358159624}
      \begin{split}
         \diver (A \nabla (\varphi u))= \sum_{i,j} \nabla^{*}_{j}\left( \varphi(x) A_{i,j}(x) 
            \nabla_{i} u  + A_{i,j}(x)\nabla_i\varphi  \otimes u(x+e_j)\right)\\
         =\sum_{j}\nabla_{j}^{*} (\varphi f_j) + \sum_{i,j} 
         A_{i,j}(x)\left(\nabla_j u(x)-f_j(x)\right) \nabla_{i} \varphi(x) + \sum_{i,j} 
         \nabla _j^{*}\left( A_{i,j} \nabla_i \varphi \otimes u(x+e_i)\right)
      \end{split}
   \end{equation}
   Denote by
   \begin{equation*}
      \begin{split}
         \tilde f_j&:= \varphi f_j    +\sum_{i} A_{i,j} \nabla_i\varphi(x)  \otimes u(x+e_i)\\
         g&:= \sum_{i,j}A_{i,j}(\nabla_j u - f_j) \nabla_{i} \varphi(x)
      \end{split}
   \end{equation*}
   Equation \eqref{eq:01141358159624}  can be rewritten as
   \begin{equation*}
      \begin{split}
         \diver (A(\varphi u)) = \diver \tilde{f} + \tilde{g} 
      \end{split}
   \end{equation*}
   Let $s=\min(p,t^{*})$.  One has that 
   \begin{equation*}
      \begin{split}
         \bigg( 
         \frac{1}{(M+\delta)^{d}}\sum_{Q_{M+\delta}}\|\tilde f\|^{s}\bigg)^{1/s}&\leq \bigg(
         \frac{1}{(M+\delta)^{d}}\sum_{Q_{M+\delta}}|\varphi f| ^{p} 
         \bigg)^{1/p} \!\!+ \sum_{i,j} \bigg( 
         \frac{1}{(M+\delta)^{d}}\sum_{Q_{M+\delta}}A_{i,j} |\nabla_i 
         \varphi|^{t^*} |u|^{t^{*}} \bigg)^{1/t^{*}}  \\ &\lesssim \bigg(
         \frac{1}{(M+\delta)^{d}}\sum_{Q_{M+\delta}}|\varphi f| ^{p} 
         \bigg)^{1/p} + 
         \bigg( 
         \frac{1}{(M+\delta)^{d}}\sum_{Q_{M+\delta}}|u|^{t^{*}} \bigg)^{1/t^{*}} 
      \end{split}
   \end{equation*}
   Using the Sobolev inequality,   the last term in the previous equation can be bounded by
   \begin{equation*}
      \begin{split}
         \left(\frac{1}{(M+\delta)^{d}}\sum_{Q_{M\delta}} 
            |u|^{t^{*}}\right)^{\frac{1}{t^{*}}} \leq 
         \left[\left(\frac{1}{(M+\delta)^d}\sum_{Q_{M + \delta }} 
               |u|^{t}\right)^{1/t}\!\! + \left(\frac{1}{(M+\delta)}\sum_{Q_{M+\delta}}|(M+\delta)\nabla u|^{t}\right)^{1/t}\right]
      \end{split}
   \end{equation*}

   In a similar way one has 

   \begin{equation*}
      \begin{split}
         \left(\frac{1}{( M +\delta )^{d}}\sum _{Q_{M+\delta}} 
            |g|^{t}\right)^{1/t} \lesssim \big( \sup_{i,j}|A_{i,j}| \big) \frac{1}{\delta} 
         \left( \frac{1}{(M+\delta)^{d}}\sum_{Q_{M+\delta}} |\nabla u|^{t} 
         \right)^{1/t}  \\ + \sup |A_{i,j}|\frac{1}{\delta} \left( 
            \frac{1}{(M+\delta)^{d}}\sum _{Q_{M+\delta}}| f_j|^{p}\right)^{p}
      \end{split}
   \end{equation*}

   Putting together all the previous inequalities and using  
   Lemma~\ref{lemma:dolzman-mueller-lemma-2}, one has that
   \begin{equation*}
      \begin{split}
         \left(\frac{1}{M^d}\sum_{Q_M} \|\nabla u\|^{s}\right)^{\frac{1}{s}} 
         \lesssim \left(\frac{1}{(M+\delta)^d}\sum_{Q_{M + \delta }} 
            |u|^{t}\right)^{1/t}  &+  \left(\frac{1}{(M+\delta)}\sum_{Q_{M+\delta}}|(M+\delta)\nabla u|^{t}\right)^{1/t}
         \\&+ \frac{M+\delta}{\delta}\left(\frac{1}{(M+\delta)^d}\sum_{Q_{M+\delta}} |f|^{p}\right)^{\frac{1}{p}}.
      \end{split}
   \end{equation*}

   Applying the previous reasoning $\kappa$ times, we have that 
   \begin{equation*}
      \begin{split}
         \left(\frac{1}{M^d}\sum_{Q_M} \|\nabla u\|^{t_\kappa}\right)^{\frac{1}{t_\kappa}} \leq C_{\kappa}
         \left(\frac{1}{(M+k\delta)^d}\sum_{Q_{M + k\delta }} 
            |u|^{t}\right)^{1/t}  &+ C_{\kappa}\left(\frac{1}{(M+k\delta)}\sum_{Q_{M+\delta}}|(M+\delta)\nabla u|^{t}\right)^{1/t}
         \\&+ C_{\kappa}\left(\frac{1}{(M+k\delta)^d}\sum_{Q_{M+k\delta}} |f|^{p}\right)^{\frac{1}{p}},
      \end{split}
   \end{equation*}
   where $t_\kappa$ is given by the recursive equation 
   $t_{j}= \max(p,t_{j-1}^{*})$ and $t_1=t$.   It can be easily seen that for 
   every $t>1$, it holds that $t_j\geq d$ for some $j$ which depends only on 
   $p$ and $q$.

\end{proof}

\begin{proposition}
   \label{prop:stima-green}
   Let $C(x,y)$ be the Green function,\ie for every $x\in \T^{d}_{N}$ one has
   \begin{equation*} 
      \begin{split}
         \diver(A\nabla C(x,\cdot ))=\delta _{x}
      \end{split}
   \end{equation*} 
   where $A$ satisfies the usual conditions.  

   Then
   \begin{equation*}
      \begin{split}
         |\nabla^{\alpha}C(x,y)|\lesssim |x-y|^{2-d-|\alpha|}.
      \end{split}
   \end{equation*}
\end{proposition}

\begin{proof}
   Let $K$ be the solution of 
   \begin{equation*}
      \begin{split}
         \diver ( \nabla K) = \delta_x.  
      \end{split}
   \end{equation*}
   It is well-known that the following estimates hold
   \begin{equation*}
      \begin{split}
         |(\nabla^{\alpha} K)(x-y)|\lesssim |x-y|^{2-d-|\alpha|}.
      \end{split}
   \end{equation*}
   From Remark~\ref{rmk:norma-debole-K}  we have that    $|(\nabla^{\alpha}K)(x-y)|_{\frac{d}{d+|\alpha| -2},\infty}\leq C_{ d,\alpha}$ where 
   $C_{ d,\alpha }$ is a constant depending only on the dimension $d$ and the multiindex $\alpha$.  

   Let us denote with $u(y)=C(x,y)$.  Then from the definitions of $K$ and $C$ 
   one has that
   \begin{equation*} 
      \begin{split}
         \diver(A\nabla u) = \diver (\nabla K(x-\cdot ))
      \end{split}
   \end{equation*}

   Let $|x-y|=R$.  Without loss of generality we may assume that $M>2m_{0}$, 
   where $m_{0}$ is the constant in Lemma~\ref{lemma:dolzmann_mueller_lemma_3}. Let $M= [\frac{R}{2}]$ and let  $Q_M$ be a cube such that 
   $y\in Q_M$ and $x\not\in Q_{ 2M }$.   Given that $\Acal C(x,\cdot)=0 $ in 
   $Q_{2M}$, using Lemma~\ref{lemma:dolzmann_mueller_lemma_3}  we have that 
   \begin{equation*}
      \begin{split}
         C(x,y)\lesssim M^{2-d} C_{d}\leq |x-y|^{2-d} C_d.  
      \end{split}
   \end{equation*}

   Higher derivative follow in a similar way. 
   For example to estimate $\nabla _{i}u$ it is enough to consider the equation
   \begin{equation*}
      \begin{split}
         \diver (A \nabla\nabla_{i} u) =  \diver ((\nabla \nabla_{i}u) 
         )- \diver ((\nabla_{i} A) \nabla u), 
      \end{split}
   \end{equation*}
   and apply the above reasoning, and hence using the global estimate one has that $|\nabla \nabla u |$

\end{proof}

\begin{proposition}
   \label{proposition:pre_kryekryesorja}
   Let $Q_1,\dots,Q_k$ be cubes of length $l_1,\cdots,l_k$ respectively such 
   that $y\in Q_i$.  Then there exists a dimensional constants $C_{d,j}$ such that
   \begin{equation}
      \label{eq:01071357515050x}
      \begin{split}
         \sup |\nabla ^{j }u|\leq 2^{k} C_{d,j}   \max\left( |x-y|,\dist(x,T^{d}_{N}\setminus Q_1),\ldots,\dist(x,\T^{d}_{N}\setminus Q_k)  \right)^{2-d+j},
      \end{split}
   \end{equation}
   where $u= (P_{Q_1}\cdots P_{Q_k} C(x,\cdot))$ and $C(x,y)$ is the Green's 
   function.  

\end{proposition}

\begin{proof}

   Let $Q_{1}$ be a cube of size $l_1$ in generic position.   
   Given that $\diver(A\nabla C_{x}(y))=0$, if $x\not\in \bar{Q}_1$ then
   $\Pi_{Q_1}C(x,y)=0$, thus $P_{Q_1}C(x,y)=C(x,y)$,
   hence the inequality follows from Proposition~\ref{prop:stima-green}.

   Let  $\varepsilon:=\dist(y, \bar{ Q }^{C}_1) <l_1 $.   
   If $|x-y|>\varepsilon/2$, then by estimating the different terms $\Pi_{Q_1}C(x,y)$ and $C(x,y)$ separately one has the desired result.  
   Indeed, it is immediate that $C(x,y)\lesssim |x-y |^{2-d}$. 
   On the other side it is not difficult to see that there exits a cube of size $\varepsilon $ touching the boundary such that it does not contain $x$ and such that twice the cube does not contain $x$. 
   Then by using Lemma~\ref{lemma:dolz-mue-orig-3}, one has that 
   \begin{equation*} 
      \begin{split}
         |\Pi _{Q_{1}} C(x,y)|\lesssim |x-y |^{2-d} M,   
      \end{split}
   \end{equation*} 
   {where}
   \begin{equation*} 
      \begin{split}
         M=\|D \Pi _{Q_1} C_{x} \|_{L^{d/d-2,\infty }(Q_{1})} + \|\Pi _{Q_1} C_{x}  \|_{L^{d/d-1,\infty }(Q_{1})}. 
      \end{split}
   \end{equation*} 

   Then by using Lemma~\ref{lemma:dolzman-mueller-lemma-2} one has that 
   \begin{equation*} 
      \begin{split}
         \|D\Pi _{B_{1}} C_x \|_{L^{d/(d-2),\infty}} + 
         \|\Pi _{B_{1}} C_x \|_{L^{d/(d-1),\infty}} \lesssim
         \|D C_x\|_{L^{d/(d-2),\infty}} + 
         \|C_x\|_{L^{d/(d-1),\infty}} 
      \end{split}
   \end{equation*} 

   Suppose that $|x-y |\leq \varepsilon /2$.   Then one can find a cube of 
   size $\lfloor \varepsilon /2\rfloor$ such that double the cube is contained 
   in $Q_{1}$. Finally by using Lemma~\ref{lemma:dolzmann_mueller_lemma_3} we 
   have the desired result.

   Let us now prove the inductive step.  Let $Q_{1},\ldots,Q_{k}$ be $k$ cubes 
   cetered in $0$.   
   If the maximum in the right hand side of \eqref{eq:01071357515050x}  is 
   $|x-y|$ or $\dist(x,\T^{d}_{n}\setminus Q_1)$, then the same reasoning  as above would apply.  
   For simplicity let us suppose that 
   \begin{equation*}
      \max\left(|x-y|,\dist(x,\T^{d}_{N}\setminus \bar Q_1),\dots,\dist(x,\T^d_N 
         \setminus \bar Q_k)\right)=\dist(x,\T^{d}_{N}\setminus \bar Q_1)=:\delta.  
   \end{equation*}
   From the inductive step we know that 
   \begin{equation*}
      \begin{split}
         \sup |v|\lesssim \delta^{2-d}\qquad  
         \sup |\nabla ^{\alpha}v|\lesssim \delta^{2-d -|\alpha|},
      \end{split}
   \end{equation*}
   where $v:=P_2\dots P_k C(x,\cdot)$.   From the definition we have that 
   $u= v- P_{Q_1}v$, hence $\sup |u| = \sup |v|+ \sup |\Pi_{Q_1} v|$.  
   Thus by using Lemma~\ref{lemma:dolzmann_mueller_lemma_3} and a very similar reasoning as above we have the desired result. 
\end{proof}

Let $Q_1,\dots,Q_k$ be $k$ cubes with radii $l_1,\dots,l_k$ respectively and 
let $\Ccal$ be the Green's function.  From now on we fix $x$  and denote with 
$u(y):= (\Rs_1 \cdots \Rs _k \Ccal(x,\cdot))(y)$, where for simplicity we will use 
$\mathcal R_{i}=\mathcal R_{Q_i}$.   

The following simple calculation will be repeatedly used  in the next theorem. 

\begin{remark}
   \label{rmk:2x}
   Let $j> 1$ be an integer and $Q$ be a cube of size $l$.  Then 
   \begin{equation*}
      \begin{split}
         \frac{1}{|Q|}\sum_{z\in Q} \max(\alpha, \dist(z,\T^d_N\setminus \bar Q))^{-j} \lesssim
         \frac{\alpha^{1-j}}{l}  
      \end{split}
   \end{equation*}

   and if $j=1$ then 
   \begin{equation*}
      \begin{split}
         \frac{1}{|Q|}\sum_{z\in Q} \max(\alpha, \dist(z,\T^d_N\setminus \bar
         Q))^{-j} \lesssim \frac{\log(\alpha)}{l}  .
      \end{split}
   \end{equation*}
   To prove the above calculation, it is enough to view it as a discretization of the Lemma~\ref{lemma:2}, hence use a similar process.
\end{remark}

\begin{theorem}
   \label{thm:diskrete-kryekryesorja}
   Let $C_k,Q_i,r_i$ as above and  such that $r_1<\dots<r_h<|x-y|<r_h+1<\dots<r_k$.  Then
   \begin{enumerate}
      \item  if $k -h< d-2$  
         \begin{equation*}
            \begin{split}
               |C_{k}(x,y)|&\lesssim \frac{1}{r_{h+1}\cdots r_k} |x-y|^{2-d+k - 
                  h}\prod_{i=h+1}^{k}\left( \log\left({|x-y|}\right)+1\right)\\
               |\nabla ^{j}_{y}C_{k}(x,y)| &\lesssim   \frac{1}{r_{h+1}\cdots 
                  r_k} |x-y|^{2-d+k -j -h}
            \end{split}
         \end{equation*}
      \item  if $k-h\geq d-2$  
         \begin{equation*}
            \begin{split}
               |C_{k}(x,y)|&\lesssim \frac{1}{ r_{k-d+3}\cdots  r_k} \left|\log( |x-y| )\right|\\
               |\nabla ^{j}_{y}C_{k}(x,y)| &\lesssim   \frac{1}{ r_{k-d +2 
                     -j}\cdots  r_k} \prod_{i=h+1}^{k}\left( \log\left({|x-y|}\right)+1\right)\\
            \end{split}
         \end{equation*}
   \end{enumerate}
\end{theorem}

\begin{proof}
   We will only prove the first part of (i). The proof of the other parts is 
   similar.

   Let us initially consider the case $k=1$.  
   For simplicity we denote $\Pi_z:= \Pi_{Q_1 +z}$.   With simple 
   computations one has
   \begin{equation*}
      \begin{split}
         \sup_{y} \left|u(y)\right| \leq \frac{1}{|Q|}\sum_{Q_1+y}\sup_{y} |(\id - \Pi_{z})u(y)|  
      \end{split}
   \end{equation*}
   Given   that for every $z\in y+Q$ it holds $\dist(y,z+Q_1)=r_1 - |z-y|$, it holds
   \begin{equation*}
      \begin{split}
         \sup |(\id - \Pi_{z})u| \leq
         \begin{cases}
            (r_1 - |z-y|)^{2-d}& \text{ if\ \  $
               r_1 - |y-z|\geq |x-y|$} \\ 
            |x-y|^{2-d}& \text{otherwise}
         \end{cases},
      \end{split}
   \end{equation*}
   The above can be reformulated as
   $\sup |(\id - \Pi_{z})u|\leq \max(|x-y|,\dist(z,\T^{d}_{N}\setminus \bar Q))$.  
   Hence using Remark~\ref{rmk:2x}  one immediately has

   \begin{equation*}
      \begin{split}
         \sup_{y} |u_1(y)| &\lesssim \frac{|x-y|^{3-d}}{r_1}.
      \end{split}
   \end{equation*}

   Let us now turn to the general case $k<d -2$.  And let $Q_1,\dots,Q_k$ be 
   balls of radiusis $r_1,\dots,r_k$ centered in $0$.   From 
   Proposition~\ref{proposition:pre_kryesorja} we have that 
   \begin{equation*}
      \begin{split}
         &\sup | P_{z_1+Q_1}\cdots P_{z_k+Q_k} C(x,\cdot)|\leq 
         \max\insieme{|x-y|,r_1 - |z_1-y|,\dots,r_k- |z_k-y|}^{2-d}\\ 
         &\leq  \max\insieme{|x-y|}^{2-d+k}\cdot\max\insieme{|x-y|,r_k- 
            |z_k-y|}^{-1}\cdots\max\insieme{|x-y|,r_k- |z_k-y|}^{-1}\\ &\hspace{10cm} =:g(z_1,\dots,z_k). 
      \end{split}
   \end{equation*}
   \begin{equation*}
      \begin{split}
         \sup \Rcal_1\cdots \Rcal_k C (x,\cdot )\leq 
         \sum_{Q_1}\cdots \sum_{Q_k} g(z_1,\ldots,z_k) 
      \end{split}
   \end{equation*}
   From Remark~\ref{rmk:2x} we have that
   \begin{equation*}
      \begin{split}
         \sum_{Q_1}\cdots \sum_{Q_k} g(z_1,\ldots,z_k)
         \leq \frac{1}{r_1\cdots r_k}|x-y|^{2-d+k}\prod_{i} (|\log (|x-y|)| 
         +1)
      \end{split}
   \end{equation*}

\end{proof}

A direct consequence is the following corrollary:
\begin{corollary}
   Suppose that $|x-y|>1$ and let $Q_1,\ldots,Q_k$ and such that $r_i=L^{i}$ 
   with $L>1$. Then there exists $\eta (j,d)$ such that  
   \begin{equation*}
      \begin{split}
         |\nabla ^{j}C_k(x,y)| \lesssim \frac{L^{\eta(j,d)}}{L^{k(d-2 -j)}}.
      \end{split}
   \end{equation*}
\end{corollary}

\begin{thm}[Fixed $A$]
   Let 
   \begin{equation}
      \label{eq:01091357722313}
      \begin{split}
         \Ccal_{k}:=\Rcal_{1}\cdots \Rcal_k \Ccal\Rcal^{*}_{k}\cdots 
         \Rcal_1^{*} - \Rcal_{1}\cdots \Rcal_{ k +1} C\Rcal^{*}_{k+1}\cdots 
         \Rcal_1^{*}.
      \end{split}
   \end{equation}  
   Then
   \begin{equation*}
      \begin{split}
         \sup_{y\in \T^{d}_{N}} |\nabla ^{\alpha}\tilde C_k(x,y) |\leq L^{\eta(d,|\alpha|)} L^{-(k-1)(d-2 + |\alpha|)}
      \end{split}
   \end{equation*}
\end{thm}
\begin{proof}
   We will estimate the two term in right hand side of \eqref{eq:01091357722313} separately. 
   Given that $\Rcal^{*}=\Acal \Rcal \Acal^{-1}$, and denoting 
   by $\Dcal_k=\Rcal_{1}\cdots \Rcal_k C\Rcal^{*}_{k}\cdots\Rcal_1^{*}$.   
   one has that
   \begin{equation}
      \begin{split}
         \Dcal_k = \Rcal_1\cdots\Rcal_{k}\Rcal_{k}\cdots \Rcal_1 C.  
      \end{split}
   \end{equation}
   Applying Theorem~\ref{thm:diskrete-kryekryesorja}, we obtain that the supremum of 
   $\Dcal_k$ is bounded by
   \begin{equation*}
      \begin{split}
         \prod_{j=1}^{d-2} L^{-k+j} \prod_{j=1}^{d-2}\log(  L^{-k+j} ) \leq  
         L^{-k(d-2)} L^{\eta(d)}.
      \end{split}
   \end{equation*}
\end{proof}

\section{Analytic dependence on $A$} 
\label{sec:analytic-dependence}

The proof of the analyticity is based on a very elegant argument using complex analysis, and it is originally found in \cite{MR2995704}.
Because most of the arguments follow by trivial modification, we will only sketch the passages. 

The main tool of the Analytic dependence is the use of the following facts:

Given an homomorphic $f:D\to C^{m\times m}$, where $D $ is the unit disk and let $M $ be such that $\sup_{z\in D} \|f(z) \| \leq M $.  Then one has that $ \|f^{j}(0) \| \leq j! M$, where $f^{j}$ is the $j $-th derivative. Moreover let $g:D\to C^{m\times m}$ be an additional homomorphic function and $\bar{M}$ such that  $ \sup_{z\in D} \|f(z) \| \leq \bar{M} $ then $\|h^{j} (0)\| \leq M \bar{M}j! $, where $h=fg^{*}$.

Fix $c_{0}$ and let $A = A_0 + z A_{1}$ such that $A_{0} $ is symmetric and such that 
\begin{equation*} 
   \begin{split}
      \scalare{A_0(x) F, F}_{\C^{m\times d}} \geq c_0  \abs{F}^2, \qquad \text{and} \qquad \sup_{x\in \T^{d}_{N}}\| A_1(x) \| \leq \frac{c_0}{2}.
   \end{split}
\end{equation*} 

As in the previous sections we define
\begin{equation*} 
   \As := \nabla^* A \nabla.
\end{equation*}
This induces the sesquilinear form $\scalare{\varphi,\psi}= \scalare{\As \varphi , \psi} $.  Notice that  if  $A$ is real and symmetric, then  $\scalare {\cdot, \cdot}_A$ is a scalar product and agrees with $\scalare{\cdot, \cdot}+$. 

One then goes on and shows that $\Ts $ defined as usual satisfies $\|\Ts_{A}\varphi \|_{A_{0}} \lesssim \|\varphi \|_{A_{0}} $.  The above fact, and the complex version Lax-Milgram theorem shows existence of the bounded inverse  $\Cs_{A}= \As^{-1}$.  Finally to conclude one shows that for every $z$ $C_{A(z),k}$ is bounded. Thus by using the complex analysis facts shown in the beginning of this section one has the desired result.

\end{document}